\documentclass{article}
  \usepackage[main,preprint]{neurips_2026}

\usepackage[utf8]{inputenc}
\usepackage[T1]{fontenc}    
\usepackage{hyperref}       
\usepackage{url}            
\usepackage{booktabs}       
\usepackage{amsfonts}       
\usepackage{nicefrac}      
\usepackage{microtype}      
\usepackage{xcolor}        

\usepackage{graphicx, amsfonts, amsmath, tikz, amssymb, float, amsthm, algorithm, algpseudocode} 
\usepackage{thmtools}
\usepackage{thm-restate}
\usepackage{placeins}

\newtheorem{theorem}{Theorem}[section]
\newtheorem{lemma}[theorem]{Lemma}

\newtheorem{definition}[theorem]{Definition}
\newtheorem{fact}[theorem]{Fact}

\theoremstyle{remark}

\newcommand{\pw}{\mathrm{pw}}
\newcommand{\dist}{\mathrm{dist}}

\usepackage[colorinlistoftodos,textsize=small,color=red!25!white,obeyFinal]{todonotes}

\title{Tree Search With Predictions}

\author{
  Michael Dinitz\thanks{Funded in part by NSF award 2228995.}\\
  Department of Computer Science\\
  Johns Hopkins University\\
  Baltimore, MD 21218 \\
  \texttt{mdinitz@cs.jhu.edu} \\
  \And
  Bob Dong \\
  Department of Computer Science\\
  Johns Hopkins University\\
  Baltimore, MD 21218 \\
  \texttt{bdong9@jh.edu} \\
}

\begin{document}

\maketitle

\begin{abstract}
    ``Algorithms with predictions'', or ``learning-augmented algorithms'', has proved to be an extremely useful paradigm for combining machine learning with traditional algorithms.  One of the textbook settings for this is searching a sorted array.  Without a prediction, classical binary search takes $O(\log n)$ queries, while with a prediction we can use ``doubling binary search'' to find the target key using $O(\log \eta)$ queries, where $\eta$ is the error of the prediction measured as the absolute value of the difference between the true location and the predicted location.  Since an array is just a path graph, in this paper we ask whether similar bounds can be achieved for search on even slightly more general graphs: trees.  We show first that the high-level answer is ``no'': there is no search algorithm that uses $O(\log \eta)$ queries, where $\eta$ is now the graph distance between the predicted location and the true location.  However, as our main result, we show that such bounds can be achieved on trees which are ``path-like'' in that they have low \emph{pathwidth}.  In particular, we prove that there is a search algorithm which uses at most $O(k \log \eta)$ queries, where $k$ is the pathwidth of the tree.  We also prove a lower bound showing that our algorithm has existentially optimal query complexity.  Finally, we show experimentally, on real-life inputs, that our algorithm has query complexity which is notably better than the simple non-prediction-based algorithm.
\end{abstract}
 
\section{Introduction}
Traditional algorithms tend to optimize for the worst case, leading to extremely impressive worst-case performance but somewhat lacking performance in the usual case.  On the other hand, approaches based on machine learning often have extremely impressive performance in the usual case (at least if the training data are close to the true data) while suffering from extremely poor worst-case behavior (e.g., if the training data comes from a completely different distribution than we see in our real application).  An important framework that attempts to get the best of both is the ``algorithms with predictions'' framework, also called ``learning-augmented algorithms'' or ``algorithms with machine-learned advice''.  In this framework we are given a problem-dependent ``prediction'' (possibly learned by a machine learning system), and want to do well if the prediction is accurate but also have good worst-case bounds if the prediction is extremely inaccurate.  See Section~\ref{sec:related} for more discussion.

A paradigmatic example of this framework is searching a sorted array.  This problem is extremely simple and has an extremely simple solution, so it is a standard warm-up and motivation for this area; it is in fact the very first example considered in the survey of the area by~\citet{MitzenmacherVassilvitskii}.  In the basic search problem we are given an array $A$ of $n$ elements, each of which is from some totally ordered universe $U$, where $A$ is in sorted order with respect to the ordering of $U$.  We are also given a target element $t$, and the goal is to find which index of $A$ contains $t$, i.e., find the value of $i$ such that $A[i] = t$.  Given $t$, let $\alpha(t)$ denote this value of $i$.  The textbook solution to this is binary search, which finds $\alpha(t)$ in $O(\log n)$ queries.  

But what if we are also given a \emph{prediction} of the correct location, i.e., we are also given an index $\hat \alpha(t)$?  Then it turns out that the folklore ``doubling binary search'' algorithm (originally proposed in an infinite search context by \citet{BENTLEY197682}) achieves query complexity $O(\log \eta)$, where $\eta = |\alpha(t) - \hat \alpha(t)|$ is the error of the prediction.  Note that  $0 \leq \eta \leq n$, so this is never (asymptotically) worse than the non-prediction $O(\log n)$ bound and, if $\eta$ is small, can be notably better.   This type of bound, which degrades in a controlled way with the prediction error, is often called a ``smoothness'' guarantee, and is the gold standard for algorithms with predictions.  This is in contrast, for example, to algorithms which only give tradeoffs between ``consistency'' (performance when the prediction error $\eta = 0$) and ``robustness'' (performance in the worst case when $\eta$ is unbounded or as large as possible).  

This framework has become extremely popular and is surprisingly powerful; see Section~\ref{sec:related} for more discussion of related work.  There is even significant work building directly from this binary search example by considering more complicated types of predictions~\citep{DILMNV24}.  However, one natural generalization of searching a sorted array has not yet been explored: searching objects more complicated than arrays.  

To motivate this, note that an equivalent formulation of the search problem on a sorted array is a search problem on \emph{paths}.  In particular, suppose that we are given a path graph $P = (V, E)$ and are trying to find a particular target node $t \in V$, and are given a prediction $s \in V$.  When we query some point $x \in V$, we are told which neighbor of $x$ is on the unique path from $x$ to $t$, i.e., we are told which direction on the path $t$ is from $x$.  Our goal is to find (query) $t$ with as few queries as possible.  It is trivial to see that this is \emph{precisely} the problem of searching a sorted array.  And interpreted in the language of graphs, we get that the doubling binary search algorithm has query complexity $O(\log \dist(s,t))$, where $\dist(s,t)$ denotes the graph distance between our prediction $s$ and the true target node $t$.  

But now by thinking of an array as a path graph, there is a natural generalization to other types of graphs.  A particularly obvious generalization is to \emph{trees}, since in trees it is always true that for any node $x$ there is a unique neighbor of $x$ that is on the (unique) path from $x$ to $t$.  So on trees we have the exact same problem, which we now define formally:

\begin{definition} \label{def:search-on-trees}
    In the \emph{Search on Trees} problem we are given a tree $T = (V, E)$, and there is an unknown target node $t \in V$.  In the prediction setting we are also given a prediction $s \in V$.  We can query any node $x \in V$, and that query will return the unique neighbor of $x$ that is on the path from $x$ to $t$.  The goal is to find (query) $t$ using as few queries as possible.  
\end{definition}

Without predictions, there is an obvious generalization of traditional binary search to trees, where instead of querying the median of the remaining interval we query the centroid of the remaining subtree.  This obviously has query complexity $O(\log n)$, like in the path/array setting~\citep{OP06}.  But what about the prediction setting?  Is there an algorithm that, similar to the array setting, only requires $O(\log \dist(s,t))$ queries?  If not, what query complexity can be achieved?  Are there at least simple classes of trees where $O(\log \dist(s,t))$ queries is possible?

\subsection{Our Results}

In this paper we answer these questions.  We first show that trees are fundamentally different than paths, in that no algorithm can achieve query complexity of $O(\log \dist(s,t))$.  However, as our main result we show that trees which are ``path-like'' \emph{do} allow for this type of query complexity.  In particular, we design such an algorithm for trees with \emph{bounded pathwidth}, a well-studied definition of how path-like a graph is.  We give theoretical bounds on its query complexity, and also show experimentally on real life data that our theory matches reality: our algorithm significantly outperforms the traditional centroid-based algorithm when our predictions are reasonably accurate on low-pathwidth trees.  Finally, we prove a matching lower bound which shows that our dependence on the pathwidth and on $\dist(s,t)$ is existentially optimal (no algorithm can have better dependence on all trees).

\subsubsection{Initial Lower Bound}

We begin with the bad news: when we generalize from paths to trees, it is no longer possible to get $O(\log \dist(s,t))$ queries.

\begin{restatable}{theorem}{BinaryTreeLB} \label{thm:lb-binary-tree}
    There is no algorithm (deterministic or randomized) for the Search on Trees problem that has expected query complexity $O(\log \dist(s,t))$ for every prediction $s \in V$ and target $t \in V$.   
\end{restatable}

This theorem turns out to be quite simple to prove.  Informally, consider the complete binary tree with $n$ leaves.  Then given a prediction $s$ which is a leaf, at least $\Omega(n)$ other leaves are at least $\Omega(\log n)$ away from $s$.  Distinguishing which is the true target must take at least $\Omega(\log n)$ queries for some target through a standard information-theoretic argument.  So in this case the number of needed queries is at least $\Omega(\dist(s,t))$, exponentially worse than our goal of $O(\log \dist(s,t))$ and essentially matching the trivial algorithm of just following the query path.  It turns out that this theorem is also a corollary of our stronger lower bound (Theorem~\ref{thm:MainLowerBound}), so rather than the direct proof sketched above we provide the proof as a corollary in Appendix~\ref{app:lower-bound}.

\subsubsection{Bounded pathwidth: upper bound}
Fortunately, we show that Theorem~\ref{thm:lb-binary-tree} is not the end of the story: many trees actually \emph{do} admit algorithms with low query complexity.  In particular, we show that trees with bounded \emph{pathwidth} admit such algorithms.  We define pathwidth formally in Section~\ref{sec:prelims}, but it is a standard way of quantifying how close a graph is to a path, in essentially the same way that treewidth measures how close a graph is to a tree.  Paths (and slight extensions to paths like caterpillars) have pathwidth $1$, while on the other extreme the complete binary tree has pathwidth $\Theta(\log n)$.  Every other tree has pathwidth somewhere between these two extremes. Note that the pathwidth is well-defined for any graph, but we will only apply it to trees.  This is not particularly unusual; there is significant previous work (which we draw on) studying the special properties of trees with low pathwidth (we will particularly use results from~\citet{doi:10.1142/S0218195904001433}).

More formally, as our main result we prove the following theorem in Section~\ref{sec:upper}.
\begin{restatable}{theorem}{MainUpperBound} \label{thm:MainUpperBound}
There is a deterministic algorithm for the Search on Trees problem which, when given a tree $T = (V, E)$ of pathwidth $k$ and a prediction $s \in V$, finds the unknown target $t \in V$ using at most $O(k \log \dist(s,t))$ queries.
\end{restatable}

To get some intuition for our algorithm, consider the special case of a \emph{caterpillar} tree.  Formally, a caterpillar is a tree in which there is a special path (usually called the \emph{spine}) such that every node not in the spine is at distance $1$ from the spine.  So a caterpillar is \emph{almost} a path, and indeed like a path has pathwidth $1$.  We claim that there is an obvious algorithm for solving the Search on Trees problem in a caterpillar.  Given a prediction $s$, let $s'$ be the closest node to $s$ on the spine (so either $s$ itself or the neighbor of $s$ that is on the spine).  Now run the classic doubling binary search algorithm for a path on the spine, starting from $s'$.  If $t$ is on the spine then this clearly takes at most $O(\log \dist(s', t)) = O(\log \dist(s, t))$ queries.  If $t$ is not on the spine, then in $O(\log \dist(s', t) -1)$ queries we will have queried the neighbor $t'$ of $t$ that is on the spine, which will then point us directly to $t$.  So in either case, we use at most $O(\log \dist(s,t))$ queries.  

Clearly this can be extended to other related graphs.  For example, we can clearly handle the (less standard) definition of a caterpillar where instead of every spine node being adjacent to a set of leaves, each spine node is adjacent to a set of arbitrarily long paths (i.e., deleting the spine results in a collection of disjoint paths).  In these types of caterpillars we simply start from $s'$ being the closest node on the spine to $s$ (not necessarily a neighbor of $s$), and then when we find $t'$ we run doubling binary search again on the path which contains $t$.  In other words, the key feature that we need is a ``spine'' path on which we can run doubling binary search, and then we can recursively search whatever (hopefully simpler) tree the target node is located in off of the spine.  It turns out that trees of bounded pathwidth are known to have a path with precisely these properties: it was shown by~\citet{doi:10.1142/S0218195904001433} that any tree of pathwidth $k$ has what they call a ``main path'' which, when deleted, leaves us with a disjoint collection of subtrees all of which have pathwidth at most $k-1$.  So we can treat this main path as a spine, doing doubling binary search on it to figure out which subtree contains the target $t$, and then recursively search the resulting $k-1$-pathwidth subtree.  Once properly formalized, it is not hard to see that we do at most $k$ doubling binary searches, each of which takes at most $O(\log \dist(s,t))$ queries, for a total query complexity of $O(k \log \dist(s,t))$ as claimed.

\paragraph{Robustness.}
It is common in the algorithms with predictions literature to focus on \emph{robustness}: if the prediction is arbitrarily bad, then we would like our algorithm to still be no worse (asymptotically) than the best algorithm without predictions.  In many problems it is highly nontrivial to get both robustness and consistency (good performance when the prediction is exactly accurate).  But we note that in the Search on Trees problem, we can get robustness ``for free'' by simply running the centroid algorithm in parallel with our algorithm.  That is, we can run each algorithm separately, and stop whenever the first one finishes.  This gives us a query complexity of $O(\min(\log n, k \log \dist(s,t)))$.  So for the remainder of the paper we do not discuss robustness.

\paragraph{Running Time.}
We focus on query complexity rather than running time, following the lead of~\citet{MitzenmacherVassilvitskii} and \citet{DILMNV24}.  There are multiple justifications for this.  First, queries might be much more expensive than computation; think of a case where each query itself requires either a significant amount of computation, or something like a physical experiment (when we do search in, for example, a scientific discovery context).  Second, our search algorithms can be computed \emph{offline}: given the tree $T$, we can compute ahead of time the policy that we will use for any given prediction $s$.  This takes time, of course, but can be done before seeing any queries.  Then when we see queries, following the policy is as simple as following a search tree, and our query time is equal to our query complexity.  So we can essentially trade increased ``preprocessing time'' for decreased ``query time''.

But even with these justifications, one might be skeptical of our algorithm if its running time were astronomical (e.g., exponential).  Fortunately, our algorithm turns out to have polynomial running time, allowing us to run relatively large-scale experiments (see Section~\ref{sec:experiments}).  The most time-intensive step is computing the hierarchy of spines, so in Appendix~\ref{app:computing k-spine} we discuss explicitly how to do this step efficiently. 

\subsubsection{Pathwidth lower bound}
A natural followup question to Theorem~\ref{thm:MainUpperBound} is whether there are matching lower bounds.  Most notably: do we actually need to have query complexity $\Omega(k \log \dist(s,t))$?  Note that Theorem~\ref{thm:lb-binary-tree} does not give such a bound, since the complete binary tree has pathwidth $\Theta(\log n)$ and requires $\Theta(\log n)$ queries on pairs at distance $\Theta(\log n)$, and so it only implies lower bounds of $\Omega(k)$ (purely as a function of $k$) or $\Omega\left(\frac{k}{\log k} \log \dist(s,t)\right)$ (as a function of both $k$ and $\log \dist(s,t)$).  

We improve this lower bound by showing through a more complicated class of trees that our upper bound is tight: no algorithm can always have query complexity $o(k \log \dist(s,t))$, even if we allow randomized algorithms and expected query complexity.

\begin{restatable}{theorem}{MainLowerBound} \label{thm:MainLowerBound}
    There is no algorithm (deterministic or randomized) for the Search on Trees problem that has expected query complexity $o(k \log dist(s,t))$ for every tree $T = (V, E)$ of pathwidth $k$, prediction $s \in V$, and target $t \in V$.
\end{restatable}

The details of this construction and proof can be found in Appendix~\ref{app:lower-bound}.  While our construction is more complicated than the simple complete binary tree, it is intuitively just a ``stretched'' version of such a tree.  The reason that the complete binary tree does not give such a lower bound is that each ``spine'' in the recursive call is just a single node (the root of the subtree), so we only need one query at every level rather than $\log \dist(s,t)$.  To get our desired lower bound, we simply have to ``stretch'' each spine at each level to be longer.  Carefully balancing parameters leads to Theorem~\ref{thm:MainLowerBound}.

\subsection{Related Work} \label{sec:related}

\paragraph{Algorithms with Predictions.} The algorithms with predictions setting has become quite popular, so there is significant work in this framework for a variety of problems and models.  We refer the interested reader to the early survey of~\citet{MitzenmacherVassilvitskii} as well as the invaluable website which keeps track of the area~\citep{ALPSweb}.  It is usually considered to have been initiated by the seminal work of~\citet{lykouris2021competitive}, who formally defined the setting and introduced the popular notions of robustness and consistency.  Since then, it has been studied in settings as diverse as online algorithms (ski rental \citep{Purohit}, scheduling \citep{LattanziLMV}, knapsack \citep{ImKQP21}, set cover \citep{BamasMS20}, and more), speeding up combinatorial algorithms~\citep{DinitzILMV21}, dynamic algorithms~\citep{BrandFNP24}, mechanism design~\citep{AgrawalBGOT24}, and many more.

Most related to our paper is the survey of~\citet{MitzenmacherVassilvitskii}, which first introduced the binary search problem in the predictions context.  Also related is the work of~\citet{DILMNV24}, which generalized the binary search problem in a different way: rather than to more general graphs (like us), they stay in the path but allow for \emph{distributional} predictions.  Extending our results to distributional predictions is an interesting problem that we leave to future work.

\paragraph{Search in trees and graphs.} As mentioned, there has been significant work extending binary search on paths to trees and even more general graphs in the non-predictions context.  This was kicked off by~\citet{OP06}, who among other results proved that the centroid algorithm (discussed earlier) has $O(\log n)$ query complexity.  A further important generalization was to the ``Search Trees on Trees'' (STT) problem, where we are given a tree (like in our settings) and are given a distribution over the nodes in that tree, and are asked to compute an optimal search strategy (equivalently, an optimal search \emph{tree}) for that tree and distribution.  Note that this is like a ``distributional prediction'', except the assumption is that the given distribution is perfectly accurate.  For example, \citet{BGKK23} showed that the centroid algorithm (mentioned earlier) not only has $O(\log n)$ query complexity, but given a distribution over the nodes the natural weighted generalization is a $2$-approximation of the \emph{optimal} search strategy and can be computed extremely quickly.  There has also been work on generalizing to graphs beyond trees~\citep{EKS16}.

\section{Preliminaries and Notation}
\label{sec:prelims}

Given a tree $T$, we denote its vertex set by $V(T)$ and its edge set by $E(T)$. For vertices $u,v\in V(T)$, let $\dist_T(u,v)$ denote the length of the unique simple path between $u$ and $v$ in $T$. If $P$ is a path in $T$ and $u,v\in V(P)$, then $\dist_P(u,v)$ denotes the number of edges on the subpath of $P$ between $u$ and $v$.

We first formalize the oracle model used throughout the paper.

\begin{definition}[Direction oracle]
\label{def:direction-oracle}
Let $T=(V(T),E(T))$ be a tree, and let $t\in V(T)$ be a hidden target vertex. The direction oracle for $t$ is the map
\[
    \mathrm{dir}_t:V(T)\to V(T)\cup\{\textsf{here}\}
\]
defined as follows. If $v=t$, then $\mathrm{dir}_t(v)=\textsf{here}$.  Otherwise, $\mathrm{dir}_t(v)$ is the unique neighbor $u$ of $v$ that lies on the simple path from $v$ to $t$ in $T$. 
An oracle query at $v$ returns $\mathrm{dir}_t(v)$.
\end{definition}

In the search problem, the algorithm is given the tree $T$ and a prediction $s\in V(T)$ for the target, but it does not know the target $t$. The algorithm may adaptively query vertices of $T$ through the direction oracle $\mathrm{dir}_t$. The algorithm succeeds when it outputs $t$; equivalently, it may stop once it queries a vertex $v$ with $\mathrm{dir}_t(v)=\textsf{here}$. The query complexity of an algorithm is the number of oracle queries it makes. When the target is clear from context, we write $\mathrm{dir}$ instead of $\mathrm{dir}_t$.

We will use the following standard path-search primitive.

\begin{fact}[Exponential search on a path]
\label{fact:exponential-search-path}
Let \(P\) be a simple path, let \(a\in V(P)\) be a starting vertex, and let \(x\in V(P)\) be an unknown target vertex. Suppose that each query to a vertex \(v\in V(P)\) returns \(\textsf{here}\) if \(v=x\), and otherwise indicates which of the two directions along \(P\) contains \(x\). Then there exists an adaptive algorithm, namely exponential (doubling binary) search, that identifies \(x\) using $O\bigl(\log(\dist_P(a,x)+1)\bigr)$ oracle queries~\citep{BENTLEY197682}.
\end{fact}

Next we recall the definition of pathwidth. The standard definition can be found in Appendix~\ref{app:prelims}.  For our purposes, we use the following definition specific for trees which was shown to be equivalent (for trees) by~\citep{KINNERSLEY1992345}. 
\begin{definition}[Recursive characterization for trees]
\label{def:pathwidth-recursive}
Let $T$ be a tree. Then $\pw(T)=0$ if and only if $T$ consists of a single vertex. For $k\ge 1$, a tree $T$ has pathwidth at most $k$ if and only if, for every vertex $x\in V(T)$, at most two connected components of $T\setminus\{x\}$ have pathwidth at least $k$.
\end{definition}

Equivalently, $\pw(T)>k$ if and only if there exists a vertex $x\in V(T)$ such that $T\setminus\{x\}$ has at least three connected components of pathwidth at least $k$.

We now define the spine structure used by the algorithm. Informally, a spine is a path whose removal lowers the pathwidth of every remaining component.

\begin{definition}[$k$-spine of a tree]
\label{def:k-spine}
Let \(T\) be a tree with \(\pw(T)=k\). A \emph{$k$-spine} of \(T\) is a path \(P\subseteq V(T)\) such that every connected component of \(T\setminus P\) has pathwidth at most \(k-1\).
\end{definition}

This is the same object as the \emph{main path} of~\citet{doi:10.1142/S0218195904001433}, who proved that every tree admits such a path. Hence every tree admits the following recursive decomposition.

\begin{definition}[$k$-spine decomposition]
\label{def:k-spine-decomposition}
Let \(T\) be a tree with \(\pw(T)=k\). A \emph{$k$-spine decomposition} is formed by recursively removing spines: first remove a $k$-spine \(K_T\) of \(T\), then recursively remove a \(\pw(C)\)-spine from each connected component \(C\) of \(T[V(T)\setminus K_T]\), stopping at singleton components. The removed spines form vertex-disjoint paths whose union is \(V(T)\).
\end{definition}

A $k$-spine decomposition need not be unique. Throughout the paper, the algorithm uses an arbitrary fixed $k$-spine decomposition; the guarantees in Section~\ref{sec:alg-analysis} hold for any such choice.

\section{Main Algorithm and Analysis} \label{sec:upper}

\subsection{The Algorithm}

We now describe the search algorithm. Fix an arbitrary $k$-spine decomposition of $T$. The algorithm maintains a current component $T_{\mathrm{cur}}$ and a current start vertex $s_{\mathrm{cur}}\in V(T_{\mathrm{cur}})$. Each phase operates on the spine of the current component. Let $P$ be the spine of $T_{\mathrm{cur}}$ specified by the fixed decomposition. The algorithm begins the phase by projecting $s_{\mathrm{cur}}$ onto $P$; that is, it chooses the closest vertex $a\in P$ to $s_{\mathrm{cur}}$.

The algorithm then performs exponential search along $P$ starting from $a$. The oracle answers on $P$ have a one-dimensional interpretation. Let $x\in P$ be the unique vertex where the path from the target to $P$ meets the spine; if the target lies on $P$, then $x$ is the target itself. For every queried vertex $v\in P$ with $v\ne x$, the oracle returns the neighbor of $v$ on $P$ in the direction of $x$. Thus, away from $x$, the search proceeds exactly as one-dimensional exponential search on a path. When the search reaches $x$, one of two things happens. If the target lies on $P$, then the oracle returns $\textsf{here}$ and the algorithm terminates. Otherwise, the oracle returns a neighbor $y\notin P$, and the target lies in the unique connected component of $T_{\mathrm{cur}}\setminus P$ containing $y$.

In the latter case, the algorithm descends into that off-spine component: it sets $T_{\mathrm{cur}}$ to the component of $T_{\mathrm{cur}}\setminus P$ containing $y$ and sets $s_{\mathrm{cur}}:=y$. Thus each phase either finds the target on the current spine or identifies the unique lower-pathwidth component containing the target. Since removing a spine reduces the pathwidth of every remaining component by at least one, there are at most $k$ nontrivial phases.

Algorithm~\ref{alg: kspine-main} gives the full procedure.

\begin{algorithm}
\caption{$k$-Spine Exponential Search}
\label{alg: kspine-main}
\begin{algorithmic}[1]
\Require Tree $T$, fixed $k$-spine decomposition, prediction $s$, direction oracle $\mathrm{dir}$
\State $T_{\mathrm{cur}} \gets T$, $s_{\mathrm{cur}} \gets s$
\While{$T_{\mathrm{cur}}$ is not a singleton}
    \State Let $P$ be the spine of $T_{\mathrm{cur}}$ in the decomposition
    \State Let $a$ be the closest vertex of $P$ to $s_{\mathrm{cur}}$
    \State Run exponential search on $P$ from $a$ until either $\textsf{here}$ is returned or the oracle points off $P$
    \If{the target is found on $P$}
        \State \Return the target
    \Else
        \State Let $x\in P$ be the vertex where the oracle points off the spine
        \State Let $y\notin P$ be the oracle answer at $x$
        \State Let $C$ be the component of $T_{\mathrm{cur}}\setminus P$ containing $y$
        \State $T_{\mathrm{cur}}\gets C$, $s_{\mathrm{cur}}\gets y$
    \EndIf
\EndWhile
\State \Return the unique vertex of $T_{\mathrm{cur}}$
\end{algorithmic}
\end{algorithm}

\subsection{Analysis}
\label{sec:alg-analysis}

We now prove that Algorithm~\ref{alg: kspine-main} satisfies Theorem~\ref{thm:MainUpperBound}. Namely, the algorithm correctly outputs the target vertex and uses at most $O(k \log \dist(s,t))$ oracle queries, where \(k=\pw(T)\).  We prove correctness in Section~\ref{sec:correctness} and query complexity in Section~\ref{sec:query-complexity}.

\subsubsection{Correctness}
\label{sec:correctness}

We first prove that Algorithm~\ref{alg: kspine-main} always returns the hidden target.

\begin{lemma}
\label{lem:correct}
Algorithm~\ref{alg: kspine-main} returns the target vertex.
\end{lemma}

\begin{proof}
Let \(t\) denote the hidden target. We prove the following invariant: at the beginning of every iteration of the while loop, $t\in V(T_{\mathrm{cur}})$.  The invariant holds initially because \(T_{\mathrm{cur}}=T\).

Now consider an arbitrary iteration and assume that \(t\in V(T_{\mathrm{cur}})\). Let \(P\) be the spine of \(T_{\mathrm{cur}}\). Since \(T_{\mathrm{cur}}\) is a tree and \(P\) is connected, there is a unique vertex \(x\in P\) closest to \(t\). Equivalently, \(x\) is the unique vertex of \(P\) at which the path from \(t\) to \(P\) meets the spine. If \(t\in P\), then \(x=t\). Otherwise, \(x\) is the unique vertex of \(P\) whose off-spine component contains \(t\).

We claim that exponential search on \(P\), started from the anchor \(a\), correctly identifies \(x\). For any queried vertex \(v\in P\), the oracle returns the neighbor of \(v\) on the unique \(v\)-to-\(t\) path, unless \(v=t\), in which case it returns \(\textsf{here}\). Thus every vertex of \(P\) on either side of \(x\) points along the spine toward \(x\). At \(x\), the oracle either returns \(\textsf{here}\), if \(x=t\), or returns the unique neighbor \(y\notin P\) on the \(x\)-to-\(t\) path. Therefore the oracle answers along \(P\) have exactly the one-dimensional structure needed by exponential search: the search either finds \(t\) on \(P\), or identifies the unique vertex \(x\in P\) where the path to \(t\) leaves the spine.

If the target is found on \(P\), the algorithm returns \(t\), so the output is correct. Otherwise, the oracle answer at \(x\) is a vertex \(y\notin P\) on the unique \(x\)-to-\(t\) path. Hence \(t\) lies in the connected component \(C\) of \(T_{\mathrm{cur}}\setminus P\) containing \(y\). The algorithm sets
\[
    T_{\mathrm{cur}}\gets C
    \qquad\text{and}\qquad
    s_{\mathrm{cur}}\gets y,
\]
so the invariant is preserved.

It remains to show that the algorithm terminates. Whenever the algorithm does not return during an iteration, it descends into a component of \(T_{\mathrm{cur}}\setminus P\). By the definition of the recursive spine decomposition, every such component has pathwidth strictly smaller than that of \(T_{\mathrm{cur}}\). Therefore each descent strictly decreases the pathwidth of the current component. Since pathwidth is a nonnegative integer, after finitely many descents the algorithm reaches a component of pathwidth \(0\), which is a singleton. By the invariant, that single vertex must be \(t\), and the final line of the algorithm returns it.
\end{proof}

\subsubsection{Query Complexity}
\label{sec:query-complexity}

We now bound the number of oracle queries made by Algorithm~\ref{alg: kspine-main}. 

\begin{theorem}
\label{thm:query-complexity}
Algorithm~\ref{alg: kspine-main} makes $O(k \log D)$ oracle queries, where \(k=\pw(T)\) and $D = \dist_T(s,t)$.
\end{theorem}

\begin{proof}
Consider an arbitrary phase of the algorithm. Let \(T_{\mathrm{cur}}\) be the current component, let \(s_{\mathrm{cur}}\) be the current start vertex, and let \(P\) be the spine of \(T_{\mathrm{cur}}\). Let \(a\) be the closest vertex of \(P\) to \(s_{\mathrm{cur}}\).

By Lemma~\ref{lem:correct}, the target \(t\) lies in \(T_{\mathrm{cur}}\) at the beginning of the phase. Let \(x\in P\) be the unique vertex of \(P\) closest to \(t\). Equivalently, \(x\) is the vertex of \(P\) at which the path from \(t\) to \(P\) meets the spine; if \(t\in P\), then \(x=t\).

The cost of the phase is the cost of exponential search on the path \(P\), starting from \(a\), until it identifies \(x\) or finds the target. By the standard analysis of exponential search on a path, this takes
\[
    O\bigl(\log(\dist_P(a,x))\bigr)
\]
queries.

We next relate \(\dist_P(a,x)\) to the original prediction error \(D\). Since \(a\) is the projection of \(s_{\mathrm{cur}}\) onto \(P\) and \(x\) is the projection of \(t\) onto \(P\), the unique \(s_{\mathrm{cur}}\)-to-\(t\) path in the tree contains the subpath of \(P\) from \(a\) to \(x\). Hence
\[
    \dist_P(a,x)
    \le
    \dist_T(s_{\mathrm{cur}},t).
\]
Moreover, \(s_{\mathrm{cur}}\) never moves farther from the target. Initially \(s_{\mathrm{cur}}=s\). Whenever the algorithm updates \(s_{\mathrm{cur}}\), it sets it to a vertex \(y\) on the current \(s_{\mathrm{cur}}\)-to-\(t\) path. Therefore
\[
    \dist_T(s_{\mathrm{cur}},t)
    \le
    \dist_T(s,t)
    =
    D.
\]
Thus every phase uses at most $O\bigl(\log(D)\bigr)$ queries.

It remains to bound the number of phases. Whenever the algorithm does not terminate during a phase, it descends into a component of \(T_{\mathrm{cur}}\setminus P\). After this update, all later queries are made inside this component, so the algorithm never returns to the removed spine \(P\). Since \(P\) is the spine of \(T_{\mathrm{cur}}\) in the fixed spine decomposition, every such component has pathwidth strictly smaller than \(\pw(T_{\mathrm{cur}})\). Thus the pathwidth of the current component decreases by at least one after every descent. Since the initial pathwidth is $k$, there are at most $k$ nontrivial spine-search phases.

Multiplying the per-phase bound by the number of phases gives $O\bigl(k\log D\bigr)$ queries in total. 
\end{proof}

\subsubsection{Proof of the Main Upper Bound}
\label{sec:main-upper-bound-proof}

We can now combine correctness and query complexity.

\MainUpperBound*

\begin{proof}
The result follows immediately from Lemma~\ref{lem:correct} and Theorem~\ref{thm:query-complexity}.
\end{proof}

\section{Experiments}
\label{sec:experiments}

We empirically evaluate the oracle-query complexity of $k$-spine search on tree instances derived from real-world networks. Since our theoretical guarantees are stated in terms of oracle queries, the main performance measure is the average number of oracle queries. Given a tree \(T=(V,E)\), a prediction \(\hat p\in V\), a target \(p\in V\), and pathwidth \(\pw(T)=k\), we compare three algorithms: centroid search, $k$-spine search, and naive trace. Centroid search is prediction-agnostic and repeatedly queries a centroid of the current feasible subtree, giving worst-case query complexity \(O(\log |V(T)|)\). $k$-spine search is Algorithm~\ref{alg: kspine-main}, implemented using a fixed $k$-spine decomposition. Naive trace starts at \(\hat p\) and follows the oracle direction until reaching \(p\), using exactly \(\dist_T(\hat p,p)+1\) queries.

Implementation details, including a constant-factor initialization heuristic, are deferred to Appendix~\ref{app:experiments}.

The relevant prediction error is $d:=\dist_T(\hat p,p)$.  The analysis of $k$-spine search suggests a leading dependence of the form \(k\log d\), while centroid search has worst-case query complexity \(O(\log |V(T)|)\). Equating these terms gives the coarse reference scale
\[
    k\log d \approx \log |V(T)|,
    \qquad\text{or equivalently}\qquad
    d\approx |V(T)|^{1/k}.
\]
This scale should not be interpreted as an exact prediction of the empirical crossing point. In the experiments, centroid search is run on the actual tree instances and sampled prediction--target pairs, and its average query cost can be smaller than the worst-case upper bound because of the structure of the tree and the sampled target distribution. Thus the main empirical question is whether $k$-spine search beats both the actual centroid baseline and the naive trace baseline over a meaningful range of prediction errors.

\subsection{Real-world tree instances}
\label{sec:experiments-real}

The input graphs are taken from the Network Repository~\citep{nr}: the Luxembourg road network, the sc-msdoor scientific-computing graph, and two Orkut social-network instances. We convert each graph to a DFS spanning tree rooted at a low-degree peripheral vertex; for directed graphs, we first take the underlying undirected largest connected component. This root is used only for tree construction, not as the prediction. For each tree, we evaluate at fixed prediction error \(d\) by uniformly sampling ordered pairs \((\hat p,p)\) with \(\dist_T(\hat p,p)=d\). We run all algorithms on the same pairs and report average oracle queries conditioned on \(d\). Additional sampling details appear in Appendix~\ref{app:experimental-details}.

Table~\ref{tab:real-world-summary} summarizes the instances and empirical crossover behavior. The column \( |V(T)|^{1/k} \) reports the coarse worst-case comparison scale above. The column ``$k$-spine wins'' records the sampled distance interval on which $k$-spine search has lower average query complexity than both baselines. The empirical crossing is the estimated distance at which the $k$-spine and centroid curves meet.

\begin{table}[htb!]
\centering
\small
\begin{tabular}{lccccc}
\toprule
Dataset & \( |V(T)| \) & $k$ & \( |V(T)|^{1/k} \) & $k$-spine wins & Crossing \\
\midrule
Luxembourg road network & \(114{,}599\) & \(5\) & \(10.28\) & \(d=10\)--\(22\) & \(23.0\) \\
sc-msdoor & \(404{,}785\) & \(6\) & \(8.60\) & \(d=10\)--\(43\) & \(43.6\) \\
soc-orkut-dir & \(3{,}072{,}441\) & \(5\) & \(19.84\) & \(d=10\)--\(53\) & \(53.7\) \\
soc-orkut & \(2{,}997{,}166\) & \(5\) & \(19.74\) & \(d=10\)--\(67\) & \(68.0\) \\
\bottomrule
\end{tabular}
\caption{Summary of the real-world DFS tree instances. Prediction--target pairs are uniformly sampled among ordered pairs at fixed tree distance \(d\). The scale \( |V(T)|^{1/k} \) is a coarse worst-case reference obtained by comparing \(k\log d\) with \(\log |V(T)|\). The win range is the sampled interval on which $k$-spine search uses fewer average oracle queries than both naive trace and centroid search.}
\label{tab:real-world-summary}
\end{table}

Figure~\ref{fig:real-world-zoom} shows the pre-crossing regime. The shaded region marks the sampled distances for which $k$-spine search has lower average oracle-query complexity than both naive trace and centroid search. Full-range plots over the entire sampled distance range are deferred to Appendix~\ref{app:full-range}.

\begin{figure*}[t]
\centering
\includegraphics[width=0.47\textwidth]{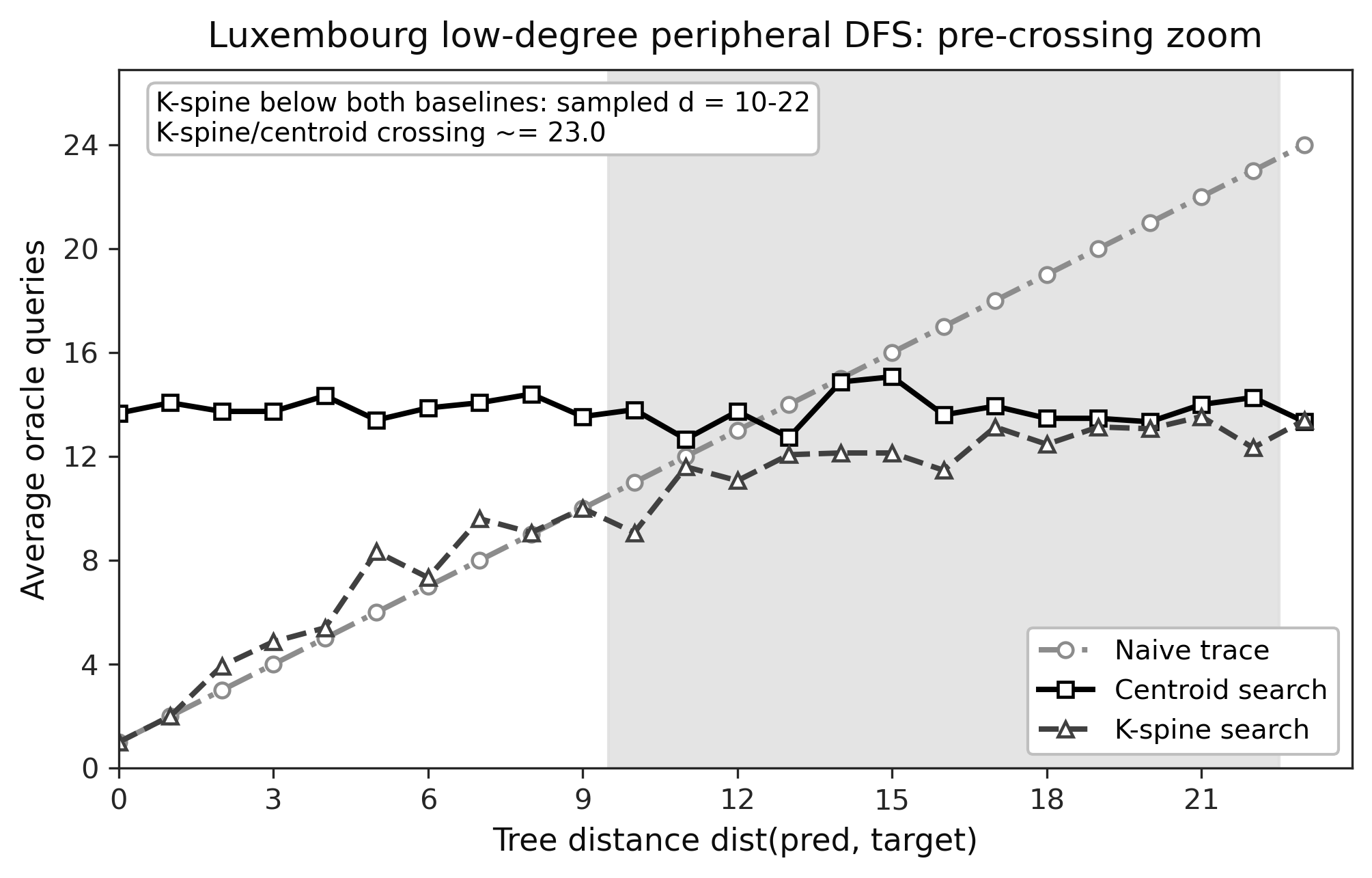}
\includegraphics[width=0.47\textwidth]{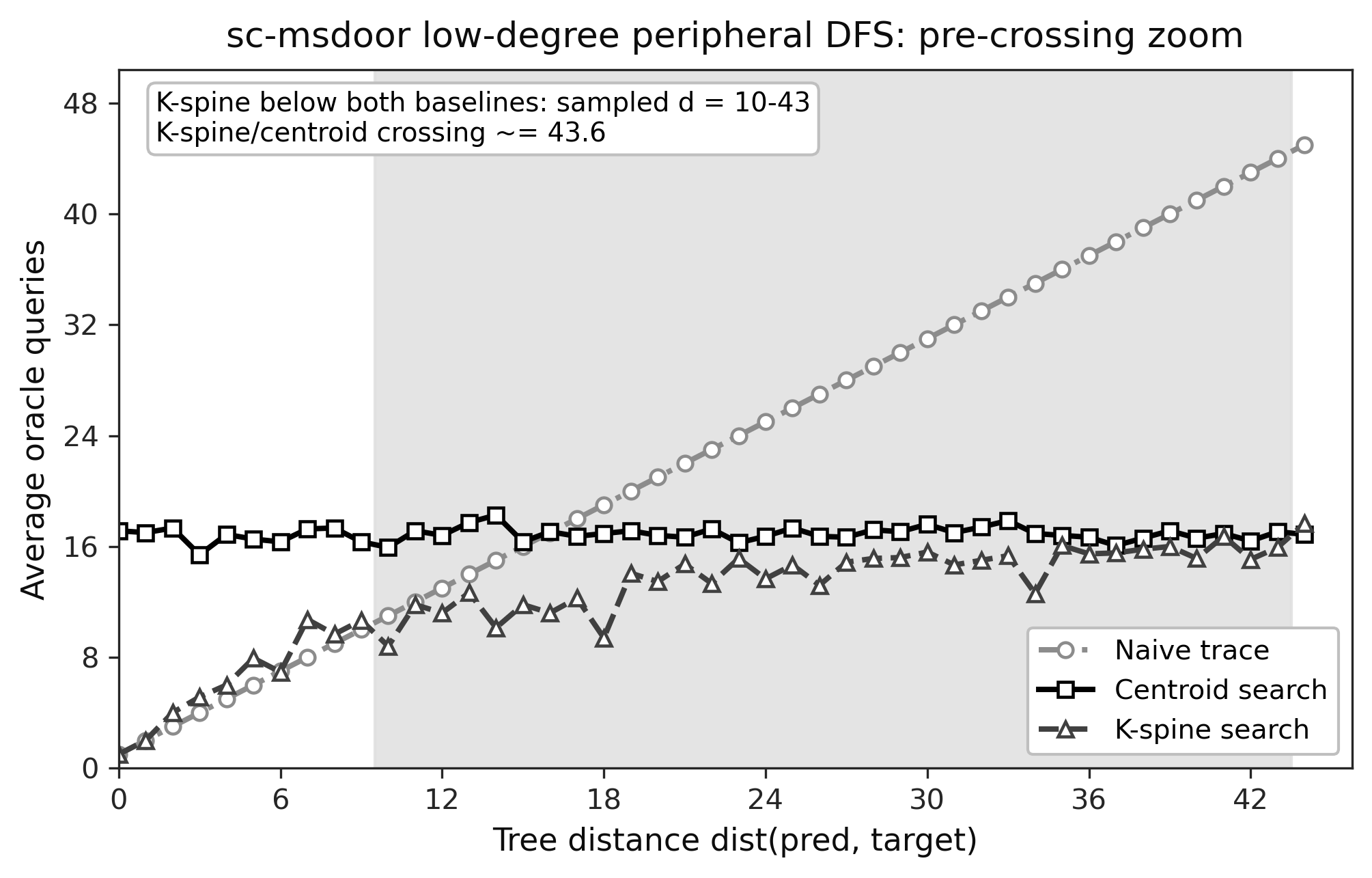}
\vspace{-0.4em}

\includegraphics[width=0.47\textwidth]{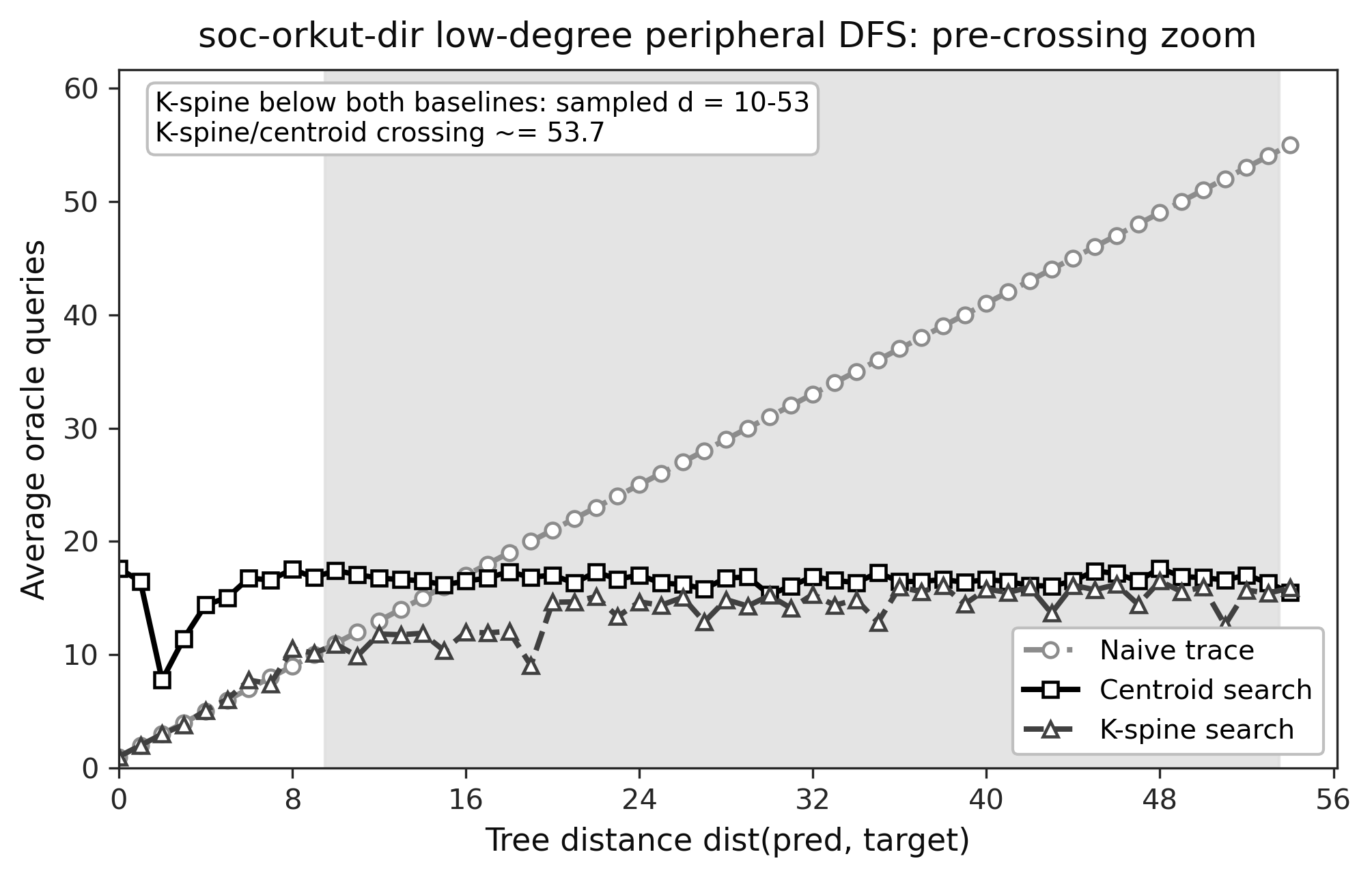}
\includegraphics[width=0.47\textwidth]{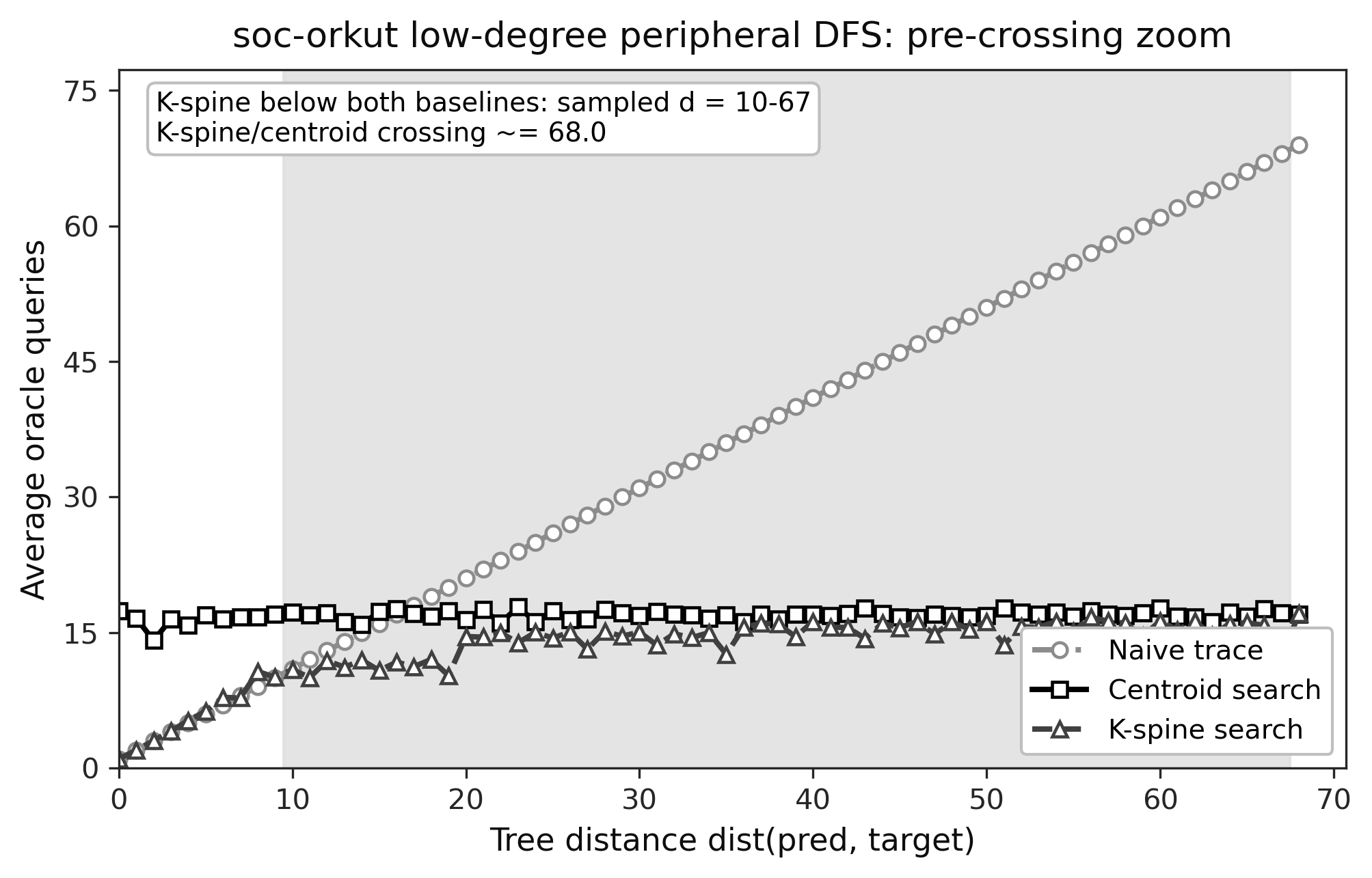}
\vspace{-0.5em}
\caption{Zoomed pre-crossing results. The shaded region marks distances where $k$-spine search uses fewer average oracle queries than both baselines.}
\label{fig:real-world-zoom}
\vspace{-0.6em}
\end{figure*}

Across all four datasets, $k$-spine search exhibits a clear intermediate-distance advantage. For very small \(d\), naive trace is difficult to beat: it has essentially no setup cost and uses exactly \(\dist_T(\hat p,p)+1\) queries, while our implementation of $k$-spine search incurs a small initial overhead before beginning the spine search. For very large \(d\), centroid search becomes competitive because its cost is essentially independent of the prediction. Between these regimes, $k$-spine search uses the prediction while avoiding the linear growth of naive trace, and it achieves the lowest average query complexity over a nontrivial interval in every instance. The initial overhead is small in all experiments and does not affect the asymptotic guarantees.

The empirical crossings are later than the coarse \( |V(T)|^{1/k} \) reference scale in all four instances: \(23.0\) versus \(10.28\) on Luxembourg, \(43.6\) versus \(8.60\) on sc-msdoor, and \(53.7,68.0\) versus roughly \(20\) on the two Orkut instances. These crossings should not be read as precise validations of a worst-case threshold, since centroid search is evaluated on the actual trees and significantly outperforms its worst-case bound in these experiments. Rather, the result is stronger than merely beating a theoretical upper bound: $k$-spine search remains below the actual centroid baseline over a substantial range of prediction errors, while also outperforming naive trace once the error is no longer very small.

\newpage

{\small
\bibliographystyle{plainnat}
\bibliography{sources}
}

\newpage
\appendix

\section{Notation and Preliminaries} \label{app:prelims}
We recall the standard definition of pathwidth~\citep{ROBERTSON198339}:

\begin{definition}[Path decomposition and pathwidth]
\label{def:path-decomposition}
Let $G=(V(G),E(G))$ be a graph. A \emph{path decomposition} of $G$ is a sequence of vertex sets
\[
    \mathcal{B}=(B_1,B_2,\dots,B_m),
    \qquad B_i\subseteq V(G),
\]
called bags, satisfying the following conditions:
\begin{enumerate}
    \item $\bigcup_{i=1}^m B_i=V(G)$;
    \item for every edge $uv\in E(G)$, there exists an index $i$ such that $\{u,v\}\subseteq B_i$;
    \item for every vertex $v\in V(G)$, the set $\{i:v\in B_i\}$ is an interval in $\{1,\dots,m\}$.
\end{enumerate}
The width of $\mathcal{B}$ is $\max_i |B_i|-1$. The \emph{pathwidth} of $G$, denoted $\pw(G)$, is the minimum width of a path decomposition of $G$.
\end{definition}

This was shown to be equivalent to Definition~\ref{def:pathwidth-recursive} through the relationship between pathwidth and vertex separation~\citep{KINNERSLEY1992345}.

\section{Lower Bounds} \label{app:lower-bound}

We now show that the dependence on the pathwidth parameter $k$ in Theorem~\ref{thm:MainUpperBound} is necessary. We construct a family of trees of pathwidth $k$ and maximum degree at most $3$ such that every deterministic or zero-error randomized algorithm has some target requiring $\Omega(k\log \dist(s,t))$ queries.

\subsection{Hard Instance Construction}
\label{subsec:lower-bound-construction}

Fix integers $k\ge 1$ and $\ell\ge \max\{4,k^2\}$. We recursively define a rooted tree $T_{h,\ell}$ for each $0\le h\le k$. The parameter $h$ denotes the remaining height of the construction, so level $0$ is the bottom level.

For the base case, let $T_{0,\ell}$ consist of a single vertex $r_0$. For $h\ge 1$, construct $T_{h,\ell}$ as follows. Start with a path
\[
    P_h=(p^h_1,p^h_2,\dots,p^h_\ell).
\]
For each $i\in\{1,\dots,\ell\}$, attach a disjoint copy $T^{(i)}_{h-1,\ell}$ of $T_{h-1,\ell}$ by adding an edge from $p^h_i$ to the root of $T^{(i)}_{h-1,\ell}$. The root of $T_{h,\ell}$ is defined to be
\[
    r_h:=p^h_1.
\]
By construction, every vertex of $T_{h,\ell}$ has degree at most $3$. The recursive structure of $T_{h,\ell}$ is illustrated in Figure~\ref{fig:hard-instance-construction}.

\begin{figure}
    \centering
    \includegraphics[width=0.9\linewidth]{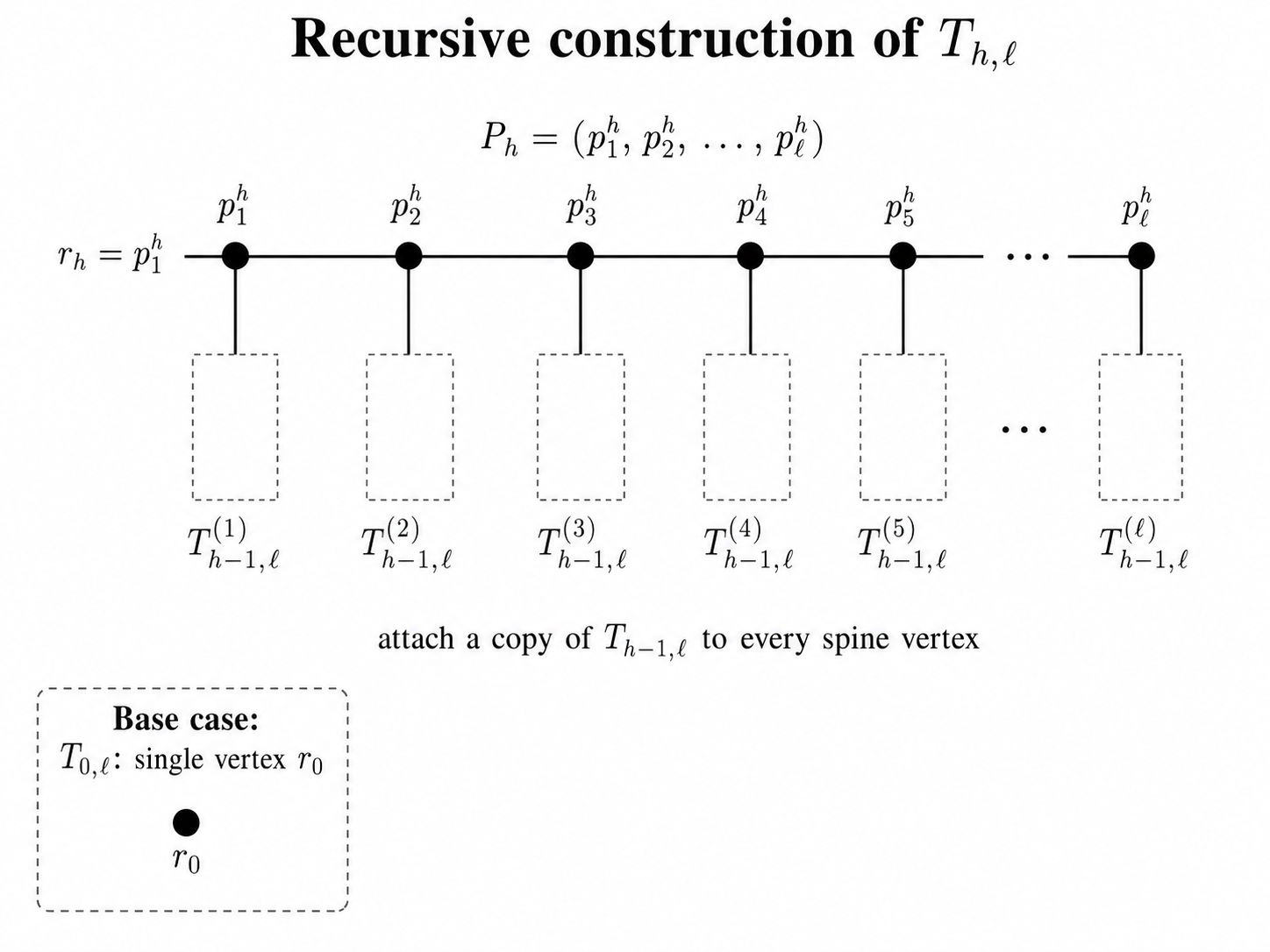}
    \caption{Recursive construction of $T_{h,\ell}$. The tree consists of a spine $P_h=(p^h_1,\dots,p^h_\ell)$, with a disjoint copy of $T_{h-1,\ell}$ attached to every spine vertex. The base case $T_{0,\ell}$ is a single vertex $r_0$.}
    \label{fig:hard-instance-construction}
\end{figure}

Next we define the bottom-level target set of the final tree $T_{k,\ell}$. Let
\[
    I:=\{\lceil \ell/2\rceil,\lceil \ell/2\rceil+1,\dots,\ell\}.
\]
For a sequence
\[
    \alpha=(i_k,i_{k-1},\dots,i_1)\in I^k,
\]
let $v_\alpha\in V(T_{k,\ell})$ be the vertex obtained by descending through the recursive construction according to $\alpha$. Concretely, starting from the root of $T_{k,\ell}$, we enter the copy of $T_{k-1,\ell}$ attached to the vertex $p^k_{i_k}$. Inside that copy, we then enter the copy of $T_{k-2,\ell}$ attached to its vertex $p^{k-1}_{i_{k-1}}$, and we continue in this way until reaching a copy of $T_{0,\ell}$. Since $T_{0,\ell}$ consists of a single vertex, this procedure identifies a unique bottom-level vertex $v_\alpha$.

We define
\[
    B_{k,\ell}:=\{v_\alpha:\alpha\in I^k\}.
\]
Thus $B_{k,\ell}$ is the set of admissible targets, all of which lie at level $0$ of the construction. Moreover,
\[
    |B_{k,\ell}|=|I|^k\ge (\ell/2)^k.
\]
Figure~\ref{fig:hard-instance-target-set} illustrates how a sequence $\alpha\in I^k$ determines a unique bottom-level target.

\begin{figure}
    \centering
    \includegraphics[width=0.9\linewidth]{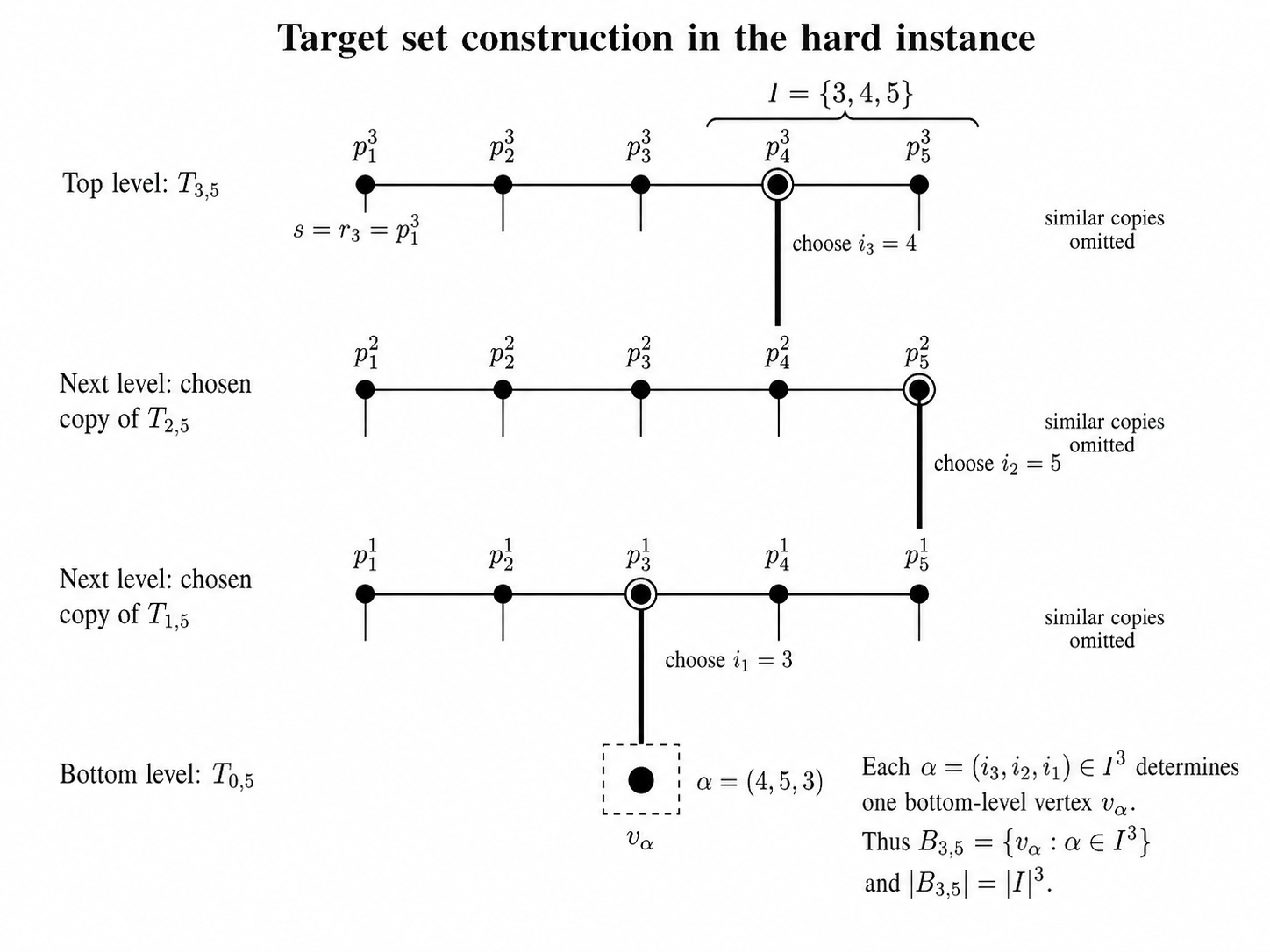}
    \caption{Illustration of the target-set construction for the hard instance, shown here for $k=3$ and $\ell=5$. At each recursive level, one chooses an attachment index in $I=\{3,4,5\}$. A sequence $\alpha=(i_3,i_2,i_1)\in I^3$ determines a unique bottom-level vertex $v_\alpha$. In general, each $\alpha\in I^k$ determines one target in $B_{k,\ell}$, so $|B_{k,\ell}|=|I|^k$.}
    \label{fig:hard-instance-target-set}
\end{figure}

The hard instance is the tree $T_{k,\ell}$ with prediction fixed to be
\[
    s:=r_k=p^k_1,
\]
the left endpoint of the top-level spine. The unknown target is promised to lie in the bottom-level target set $B_{k,\ell}$.

Intuitively, $T_{h,\ell}$ consists of a spine on $\ell$ vertices, with a copy of $T_{h-1,\ell}$ attached to every spine vertex. A target is specified by making one attachment choice from $I$ at each of the $k$ recursive levels. After $k$ such choices, the process reaches a bottom-level copy of $T_{0,\ell}$, which consists of a single vertex. Hence there are $\Theta(\ell)$ choices at each of $k$ levels, and therefore
\[
    |B_{k,\ell}|=|I|^k=\Theta(\ell^k).
\]

We first verify that this construction has the desired pathwidth.

\begin{lemma}[Pathwidth of the construction]
\label{lem:lb-pathwidth}
For every $h\ge 0$, the tree $T_{h,\ell}$ has pathwidth exactly $h$.
\end{lemma}

\begin{proof}
We prove the claim by induction on $h$. The case $h=0$ is immediate, since $T_{0,\ell}$ is a single vertex.

First, we prove the upper bound. Assume inductively that $\pw(T_{h-1,\ell})=h-1$. For each attached copy $T^{(i)}_{h-1,\ell}$, take a path decomposition of width $h-1$ and add the attachment vertex $p^h_i$ to every bag. This increases the width by at most one, so the resulting bags have width at most $h$. Moreover, these bags cover the attachment edge from $p^h_i$ to the root of $T^{(i)}_{h-1,\ell}$.

Let $\mathcal{D}_i$ denote this modified path decomposition of the $i$th attached copy. Concatenate the decompositions in the order
\[
    \mathcal{D}_1,\ \{p^h_1,p^h_2\},\ \mathcal{D}_2,\ \{p^h_2,p^h_3\},\ \dots,\ \{p^h_{\ell-1},p^h_\ell\},\ \mathcal{D}_\ell.
\]
This is a valid path decomposition of $T_{h,\ell}$. Every edge inside an attached copy is covered by the corresponding $\mathcal{D}_i$, every attachment edge is covered because $p^h_i$ was added to all bags of $\mathcal{D}_i$, and every edge of $P_h$ is covered by one of the two-vertex bags. The contiguity condition also holds: each vertex inside an attached copy appears only inside its copy's decomposition, while each spine vertex $p^h_i$ appears only in $\mathcal{D}_i$ and in the adjacent path-edge bags. Therefore,
\[
    \pw(T_{h,\ell})\le h.
\]

It remains to prove the lower bound. Choose an internal vertex $p^h_i$ of $P_h$, with $2\le i\le \ell-1$. Removing $p^h_i$ creates three connected components: one containing the left side of $P_h$, one containing the right side of $P_h$, and the attached copy $T^{(i)}_{h-1,\ell}$. Each of these three components contains a copy of $T_{h-1,\ell}$, and hence each has pathwidth at least $h-1$ by the induction hypothesis and monotonicity of pathwidth under taking subgraphs.

We use the standard three-component obstruction for pathwidth on trees: if deleting a vertex of a tree leaves at least three connected components of pathwidth at least $r$, then the tree has pathwidth at least $r+1$. This follows from the recursive characterization of pathwidth on trees, equivalently from the vertex-separation characterization~\citep{KINNERSLEY1992345,ELLIS199450}. Applying this with $r=h-1$ gives
\[
    \pw(T_{h,\ell})\ge h.
\]
Combining the upper and lower bounds gives $\pw(T_{h,\ell})=h$.
\end{proof}

We next relate $\ell$ to the distance between the prediction and every possible target.

\begin{lemma}[Prediction error in the construction]
\label{lem:lb-distance}
For every target $t\in B_{k,\ell}$,
\[
    \dist_{T_{k,\ell}}(s,t)=\Theta(k\ell).
\]
Consequently, since $\ell\ge k^2$,
\[
    \log \dist_{T_{k,\ell}}(s,t)=\Theta(\log \ell).
\]
\end{lemma}

\begin{proof}
By the definition of $B_{k,\ell}$, every target $t\in B_{k,\ell}$ is a bottom-level vertex indexed by a sequence
\[
    \alpha=(i_k,i_{k-1},\dots,i_1)\in I^k.
\]
At level $j$, the path from the root of the current copy of $T_{j,\ell}$ to the next lower-level copy moves along the current spine from $p^j_1$ to $p^j_{i_j}$ and then takes the attachment edge into the chosen copy of $T_{j-1,\ell}$. The distance along the spine is $i_j-1$, and the attachment edge contributes one additional edge, so level $j$ contributes exactly $i_j$ edges. Therefore,
\[
    \dist_{T_{k,\ell}}(s,t)=\sum_{j=1}^k i_j.
\]
Since each $i_j\in I$, we have $\ell/2\le i_j\le \ell$. Hence
\[
    \frac{k\ell}{2}
    \le
    \dist_{T_{k,\ell}}(s,t)
    \le
    k\ell.
\]
Thus $\dist_{T_{k,\ell}}(s,t)=\Theta(k\ell)$. Since $\ell\ge k^2$, we have $k\le \sqrt{\ell}$, and therefore
\[
    \log \ell
    \le
    \log(k\ell)
    \le
    \log(\ell^{3/2})
    =
    \frac{3}{2}\log \ell.
\]
Thus $\log(k\ell)=\Theta(\log \ell)$, which proves the second claim.
\end{proof}

We now prove a distributional decision-tree lower bound for the uniform distribution over the bottom-level targets.

\begin{lemma}[Distributional decision-tree lower bound]
\label{lem:lb-distributional}
Let $\mu$ be the uniform distribution over $B_{k,\ell}$. For every deterministic correct search algorithm $A$,
\[
    \mathbb{E}_{t\sim\mu}[Q_A(t)]
    =
    \Omega(k\log \ell),
\]
where $Q_A(t)$ denotes the number of oracle queries made by $A$ when the target is $t$.
\end{lemma}

\begin{proof}
Let
\[
    N:=|B_{k,\ell}|.
\]
Restricted to targets in $B_{k,\ell}$, the deterministic algorithm $A$ induces a decision tree. Since $T_{k,\ell}$ has maximum degree at most $3$, each oracle query has at most four possible outcomes: the answer $\textsf{here}$ or one of at most three neighboring vertices. Hence this decision tree has branching factor at most $4$.

For each target $t\in B_{k,\ell}$, let $d_t$ be the depth of the leaf reached by $A$ on target $t$. Since $A$ is correct on every target, distinct targets in $B_{k,\ell}$ must reach distinct leaves; otherwise the same transcript would force the same output on two different targets. Moreover,
\[
    Q_A(t)=d_t.
\]

We lower bound the average leaf depth by a simple counting argument. The case $N<16$ is absorbed into the constant in the $\Omega(\cdot)$ notation, so assume $N\ge 16$. Since the decision tree has branching factor at most $4$, the number of leaves of depth at most $q$ is at most
\[
    \sum_{i=0}^{q}4^i
    =
    \frac{4^{q+1}-1}{3}
    \le
    4^{q+1}.
\]
Choose
\[
    q:=\left\lfloor \log_4 N \right\rfloor-2.
\]
Then
\[
    4^{q+1}
    \le
    \frac{N}{4}.
\]
Thus at most $N/4$ targets in $B_{k,\ell}$ can reach leaves of depth at most $q$. Consequently, at least $3N/4$ targets have leaf depth greater than $q$.

Since $\mu$ is the uniform distribution on $B_{k,\ell}$, we obtain
\[
    \mathbb{E}_{t\sim\mu}[Q_A(t)]
    =
    \frac{1}{N}\sum_{t\in B_{k,\ell}} d_t
    \ge
    \frac{3}{4}q
    =
    \Omega(\log N).
\]
Using $N=|B_{k,\ell}|=|I|^k\ge(\ell/2)^k$, this gives
\[
    \mathbb{E}_{t\sim\mu}[Q_A(t)]
    =
    \Omega(\log |B_{k,\ell}|)
    =
    \Omega(k\log \ell),
\]
where the last equality uses $\ell\ge 4$.
\end{proof}

We are now ready to prove the lower bound.

\MainLowerBound*

\begin{proof}
Fix $k\ge 1$ and $\ell\ge \max\{4,k^2\}$, and consider the hard instance $T_{k,\ell}$ with prediction $s=r_k$. By Lemma~\ref{lem:lb-pathwidth}, this tree has pathwidth $k$.

For deterministic algorithms, Lemma~\ref{lem:lb-distributional} implies that under the uniform distribution on $B_{k,\ell}$, the average query cost is $\Omega(k\log \ell)$. Therefore, for every deterministic correct search algorithm $A$, there exists some target $t\in B_{k,\ell}$ on which $A$ makes $\Omega(k\log \ell)$ queries.

For randomized algorithms, we assume the algorithm must always output the correct target, and its query complexity is the expected number of oracle queries over its internal randomness. By Yao's minimax principle~\citep{4567946},
\[
    \min_{\mathcal{R}}
    \max_{t\in B_{k,\ell}}
    \mathbb{E}_{\rho}\bigl[Q_{\mathcal{R}_{\rho}}(t)\bigr]
    \ge
    \min_A
    \mathbb{E}_{t\sim\mu}\bigl[Q_A(t)\bigr],
\]
where $\mathcal{R}$ ranges over zero-error randomized algorithms, $\rho$ denotes the internal randomness of $\mathcal{R}$, and $A$ ranges over deterministic correct search algorithms. By Lemma~\ref{lem:lb-distributional}, the right-hand side is $\Omega(k\log \ell)$. Hence, for every zero-error randomized algorithm $\mathcal{R}$, there exists a target $t\in B_{k,\ell}$ such that
\[
    \mathbb{E}_{\rho}\bigl[Q_{\mathcal{R}_{\rho}}(t)\bigr]
    =
    \Omega(k\log \ell),
\]
with the tree $T_{k,\ell}$ and prediction $s=r_k$ fixed.

Finally, by Lemma~\ref{lem:lb-distance}, every target $t\in B_{k,\ell}$ satisfies
\[
    \log \dist_{T_{k,\ell}}(s,t)=\Theta(\log \ell).
\]
Thus the hard target guaranteed above satisfies
\[
    \mathbb{E}_{\rho}\bigl[Q_{\mathcal{R}_{\rho}}(t)\bigr]
    =
    \Omega\bigl(k\log \dist_{T_{k,\ell}}(s,t)\bigr).
\]
The deterministic statement follows identically without the expectation over $\rho$. Hence no deterministic or zero-error randomized algorithm can guarantee $o(k\log \dist(s,t))$ expected query complexity on every tree of pathwidth $k$, prediction $s$, and target $t$.
\end{proof}

Finally, we show that \autoref{thm:MainLowerBound} implies the statement of \autoref{thm:lb-binary-tree}.

\BinaryTreeLB*
\begin{proof}
Suppose, for contradiction, that there exists an algorithm for Search on Trees with expected query complexity
\[
    O(\log \dist(s,t))
\]
for every tree \(T\), every prediction \(s\in V(T)\), and every target \(t\in V(T)\). In particular, this guarantee would hold on the hard instances \(T_{k,\ell}\) from Section~\ref{subsec:lower-bound-construction}, for every choice of $k$.

Choose any non-constant sequence of pathwidths \(k\to\infty\). On the corresponding hard instances, the assumed guarantee gives
\[
    O(\log \dist(s,t))
    =
    o\bigl(k\log \dist(s,t)\bigr).
\]
This contradicts \autoref{thm:MainLowerBound}, which states that no deterministic or zero-error randomized algorithm can guarantee \(o(k\log \dist(s,t))\) expected query complexity on all trees of pathwidth $k$.

Therefore, no universal \(O(\log \dist(s,t))\) expected-query guarantee is possible for Search on Trees.
\end{proof}

\section{Experiments} \label{app:experiments}
We use one implementation optimization that improves constants but does not affect the asymptotic guarantees. Instead of always starting from the top-level spine, the algorithm first queries the prediction \(\hat p\). If the oracle returns \(\textsf{here}\), the algorithm terminates. Otherwise, the oracle answer identifies the neighbor of \(\hat p\) on the unique path to the target, and hence certifies which side of \(\hat p\) contains the target.

The implementation then starts the spine-search procedure from the lowest decomposition component consistent with this certified direction, rather than from the root of the decomposition. Intuitively, this skips high-level components that cannot contain the target. In our implementation, this component is found by searching upward through the decomposition hierarchy. This optimization can reduce unnecessary high-level spine searches in practice, but it does not affect correctness: the algorithm only skips components certified not to contain the target. It also does not change the asymptotic query-complexity guarantees of Algorithm~\ref{alg: kspine-main}.

The optimization introduces a small initial overhead because the prediction is queried before the recursive spine procedure begins. This explains why, at very small prediction error, naive trace can be slightly better: naive trace has almost no setup cost and uses exactly \(\dist_T(\hat p,p)+1\) oracle queries. The overhead is small in all experiments and is dominated once the prediction error leaves this very local regime.

\subsection{Experimental Details} \label{app:experimental-details}
For each resulting tree \(T=(V,E)\), we evaluate performance at fixed tree distance. For a distance \(d\), define
\[
    \mathcal{P}_d
    :=
    \{(\hat p,p)\in V(T)\times V(T): \dist_T(\hat p,p)=d\}.
\]
We uniformly sample ordered pairs \((\hat p,p)\) from \(\mathcal{P}_d\), treat \(\hat p\) as the prediction and \(p\) as the target, and run all three algorithms on the same sampled pairs. Thus the \(x\)-axis in the plots is the prediction error \(d\), and the \(y\)-axis is the average number of oracle queries conditioned on that error.

\subsection{Full-range plots} \label{app:full-range}

Figure~\ref{fig:real-world-full-log} shows the full-range results over the sampled distance range. These plots use a logarithmic \(x\)-axis to display behavior across both small and large prediction errors. They illustrate the large-scale qualitative behavior: naive trace grows with the prediction error, while centroid search and $k$-spine search remain substantially flatter.

\begin{figure*}[htb!]
\centering
\includegraphics[width=0.49\textwidth]{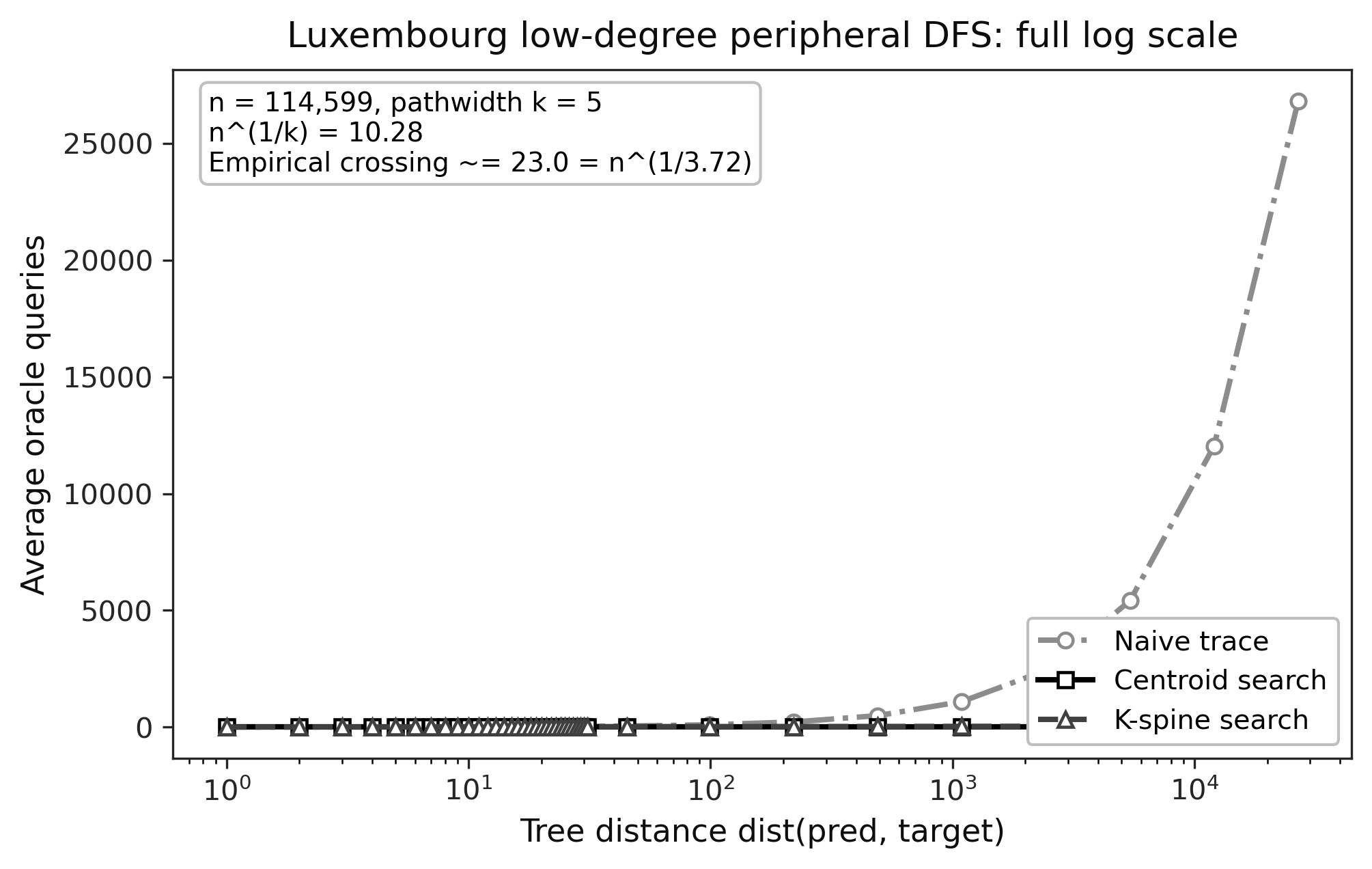}
\includegraphics[width=0.49\textwidth]{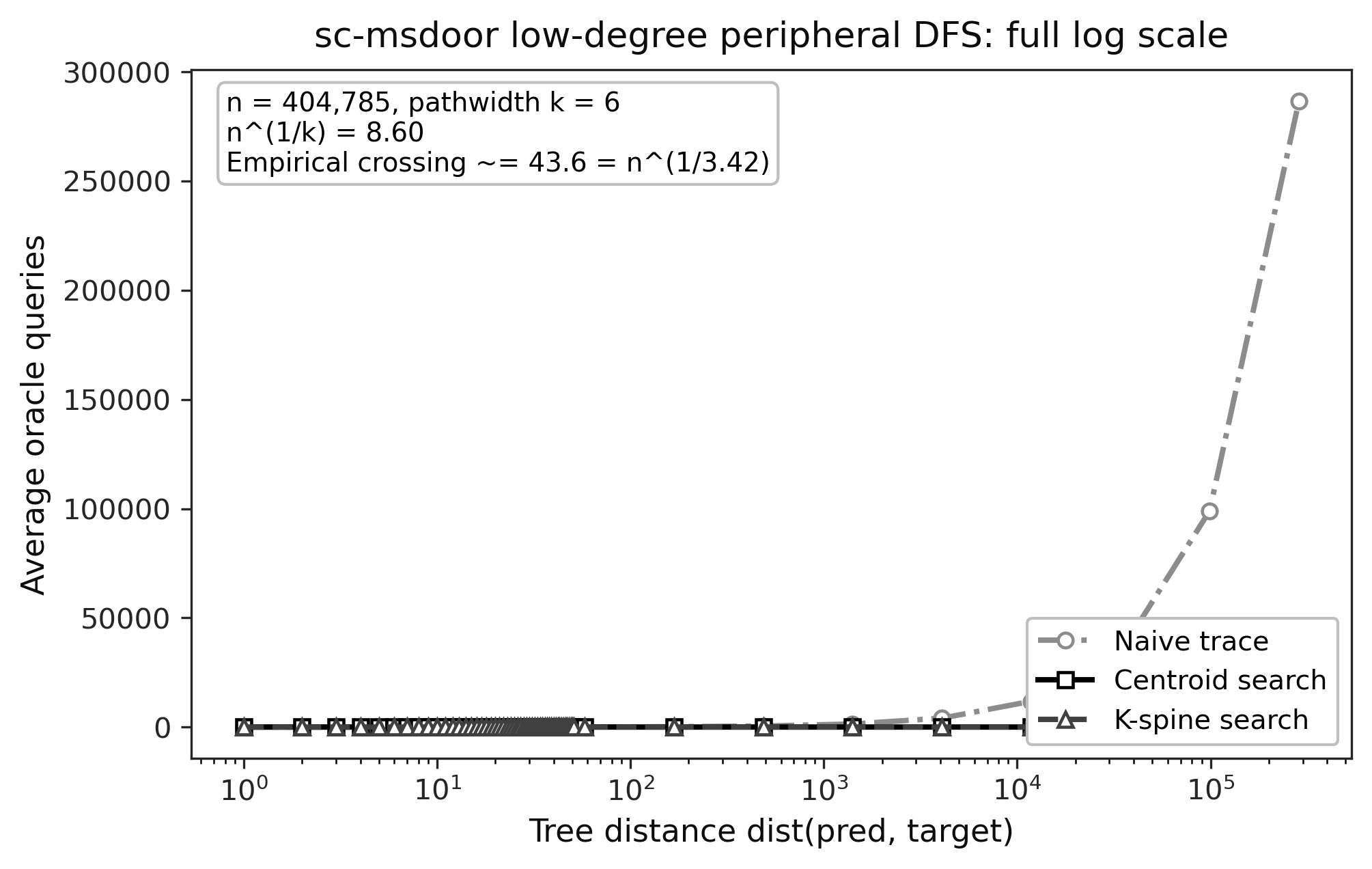}

\vspace{0.4em}

\includegraphics[width=0.49\textwidth]{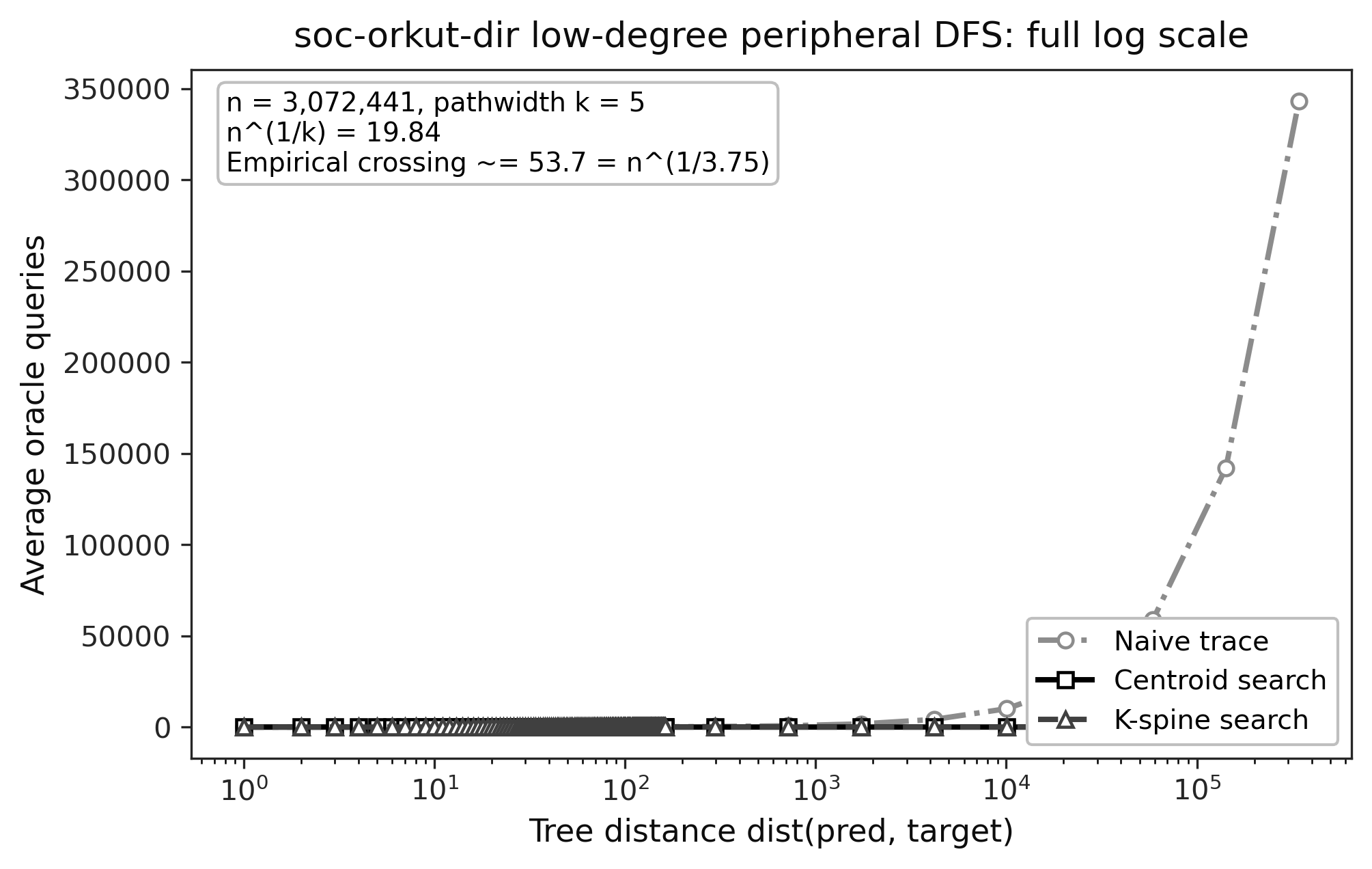}
\includegraphics[width=0.49\textwidth]{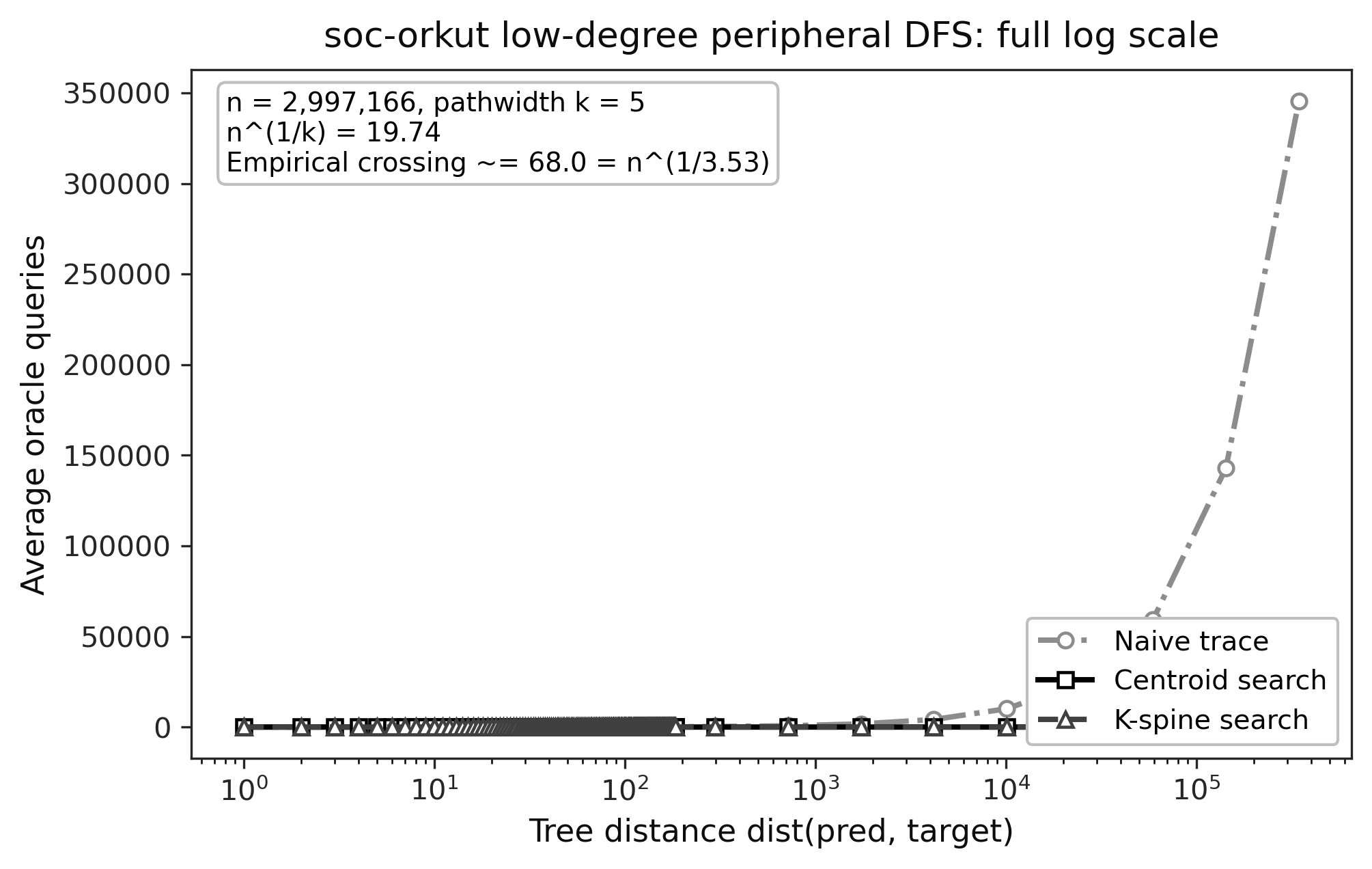}
\caption{Full-range results on real-world DFS tree instances. Top row: Luxembourg road network and sc-msdoor. Bottom row: soc-orkut-dir and soc-orkut. Naive trace grows with the prediction error, while centroid search and $k$-spine search remain substantially flatter on this scale.}
\label{fig:real-world-full-log}
\end{figure*}

\section{Computing $k$-Spine}
\label{app:computing k-spine}

We briefly describe the preprocessing step used to compute the recursive spine decomposition. This step is not central to our contribution: our main results concern the number of oracle queries once such a decomposition is available. We therefore use a standard exact tree pathwidth routine as a preprocessing subroutine, rather than optimizing this part of the implementation.

Our implementation is based on the classical characterization of tree pathwidth via vertex separation, together with the EST labeling procedure for trees~\citep{ELLIS199450}. Since vertex separation is equivalent to pathwidth, the EST labels allow us to compute the exact pathwidth of each tree component. For each connected component in the recursive decomposition, we root the component at a chosen vertex and compute EST labels bottom-up. These labels identify the bottleneck pathwidth levels in the rooted subtrees and provide enough information to recover a canonical backbone path. We use this backbone as the current spine. Removing this spine leaves connected components of strictly smaller pathwidth, and we recursively apply the same procedure to each remaining component.

Equivalently, one can view this as an implementation of the following straightforward exact procedure. For a connected subtree $T'$, compute its pathwidth level $k'=\pw(T')$. Then find a path $P$ such that every connected component of $T'\setminus P$ has pathwidth at most $k'-1$. The path $P$ becomes the spine at the current node of the decomposition tree, and the algorithm recurses on the components of $T'\setminus P$.

The preprocessing runs in polynomial time. At each recursive level, the current components are disjoint subtrees of the original tree, so the total number of vertices processed at that level is at most $n$. Moreover, the recursion depth is at most $\pw(T)$, since the pathwidth level decreases by at least one after each spine removal. For trees, $\pw(T)=O(\log |V|)$, so there are at most logarithmically many recursive levels. Since the underlying EST vertex-separation routine runs in polynomial time on trees, the full preprocessing procedure is polynomial-time as well.

We emphasize that this preprocessing is used only to construct the decomposition before running the search algorithm. The focus of the paper is the query complexity of prediction-based search on the resulting decomposition, not the optimization of the decomposition routine itself.

\section{Experimental Setup}
All experiments are conducted in Python 3.13.9 on a 16-inch MacBook Pro equipped with an Apple M2 Pro chip, consisting of a 12-core CPU with 8 performance cores and 4 efficiency cores, 16 GB of unified memory, and 512 GB of storage, running macOS Tahoe version 26.4.1. We emphasize that our primary metric is oracle-query complexity, which is hardware-independent; the hardware specification is reported only for reproducibility of the implementation-level experiments.

\newpage

\end{document}